\documentclass[journal]{IEEEtran}
\ifCLASSINFOpdf
\else
\fi
\usepackage{amsmath}
\usepackage{graphicx}
\usepackage{amssymb}
\usepackage{colortbl}
\usepackage{arydshln}
\usepackage{stfloats}
\usepackage[vlined,boxed,ruled,linesnumbered]{algorithm2e}
\SetKwProg{Fn}{Function}{}{}
\usepackage{algpseudocode}
\usepackage{multirow}
\usepackage{extarrows}
\usepackage{mathrsfs}
\usepackage{upgreek}
\newtheorem{problem}{Problem}[section]

\newtheorem{assumption}{Assumption}[section]
\newtheorem{definition}{Definition}[section]
\newtheorem{theorem}{Theorem}[section]
\newtheorem{lemma}{Lemma}[section]
\newtheorem{proposition}{Proposition}[section]

\begin{document}
%
\title{ Reachable Set Estimation for Neural Network Control Systems: A Simulation-Guided Approach}
%
%
%

\author{Weiming~Xiang,~\IEEEmembership{Senior Member,~IEEE,}
        Hoang-Dung~Tran, Xiaodong Yang and Taylor T. Johnson
\thanks{W. Xiang is with the School of Computer and Cyber Sciences, Augusta University, Augusta
GA 30912, USA. Email: wxiang@augusta.edu.}
\thanks{H.-D. Tran, X. Yang and T. T. Johnson are with Department of Electrical Engineering and Computer Science, Vanderbilt University, Nashville TN 37212, USA. Emails: dung.h.tran@vanderbilt.edu (H.-D. Tran), xiaodong.yang@vanderbilt.edu (X. Yang), taylor.johnson@vanderbilt.edu (T. T. Johnson) }
}

%
%

\markboth{IEEE Transactions on Neural Networks and Learning Systems,~Vol.~00, No.~0, XX~XXXX}%
{Shell \MakeLowercase{\textit{et al.}}: Bare Demo of IEEEtran.cls for IEEE Journals}
%



\maketitle

\begin{abstract}
The vulnerability of artificial intelligence (AI) and machine learning (ML) against adversarial disturbances and attacks significantly restricts their applicability in safety-critical systems including cyber-physical systems (CPS) equipped with neural network components at various stages of sensing and control.
This paper addresses the reachable set estimation and safety verification problems for dynamical systems embedded with neural network components serving as feedback controllers. The closed-loop system can be abstracted in the form of a continuous-time sampled-data system under the control of a neural network controller. First, a novel reachable set computation method in adaptation to simulations generated out of neural networks is developed. The reachability analysis of a class of feedforward neural networks called multilayer perceptrons (MLP) with general activation functions is performed in the framework of interval arithmetic. Then, in combination with reachability methods developed for various dynamical system classes modeled by ordinary differential equations, a recursive algorithm is developed for over-approximating the reachable set of the closed-loop system. The safety verification for neural network control systems can be performed by examining the emptiness of the intersection between the over-approximation of reachable sets and unsafe sets. The effectiveness of the proposed approach has been validated with evaluations on a robotic arm model and an adaptive cruise control system.
\end{abstract}


%
\IEEEpeerreviewmaketitle

\section{Introduction}
Neural networks have  been demonstrated to be effective tools in controlling complex systems in a variety of research activities such as stabilization \cite{wu2014exponential,yu1998stable}, adaptive control \cite{ge1999adaptive,hunt1992neural}. In some latest applications, neural networks have been deployed and played a critical role in high-safety-assurance systems such as autonomous systems \cite{julian2017neural}, unmanned vehicles \cite{bojarski2016end} and aircraft collision avoidance systems \cite{julian2016policy}. However, due to the vulnerability neural networks against adversarial disturbances/attacks and the black-box nature of neural networks, such controllers with neural network structure, in essence, are only restricted to the control applications with the lowest levels of requirements of safety as there is a short of effective methods to compute the output reachable set of neural networks and further assure the safety of underlying closed-loop systems. It has been frequently observed that even a slight perturbation against the input of a well-trained neural network will produce a completely incorrect and unpredictable output \cite{szegedy2013intriguing}. As we consider a closed-loop system with a feedback channel involving neural networks, the safety issues will inevitably arise since disturbances and uncertainties are unavoidable in measurement and control channels, which may result in undesirable and unsafe system behaviors even instability. Furthermore, with advanced adversarial machine learning techniques developed recently, such safety matters for safety-critical control systems with neural network controllers only become even much worse. Therefore, to integrate AI/ML components such as neural networks into safety-critical control systems, safety verification for such AI/ML systems is required at all stages for the purpose of safety assurance. However, because of the sensitivity of neural networks against perturbations and the complex structure of neural networks, the verification of neural networks represents extreme difficulties. It has been demonstrated that a simple property verification of a small scale neural networks is a non-deterministic polynomial (NP) complete problems \cite{katz2017reluplex}.

In recent years, there are a few methods developed for the verification of neural networks. A simulation-based approach was developed in \cite{xiang2018output} to convert the output reachable set computation problem of a feedforward neural network into a sequence of four convex optimization problems utilizing the concept of maximal sensitivity.  This paper will particularly focus on improving the simulation-based approach developed in  \cite{xiang2018output} for the output set over-approximation of feedforward neural networks with general activation functions. A novel adaptive simulation-guided method will be developed and further integrated for safety verification of closed-loop systems with neural network controllers.

\subsection{Related Work} 

Formal verification of neural networks has been well-recognized in recent literature.  One of the earliest methods is the abstraction-refinement approach proposed in \cite{pulina2010abstraction,pulina2012challenging}, which is developed for computing the output set of a feedforward neural network to perform verification. In \cite{huang2017safety}, a satisfiability modulo theories (SMT) solver was proposed for the verification of feedforward neural networks. Some Lyapunov function based approaches were proposed for dynamical systems with neural network structures \cite{xu2016reachable,zuo2014non,thuan2018reachable}, in which invariant sets are constructed to estimate reachable sets. 
For a special class of neural networks with rectified linear unit (ReLU) neurons, several methods have been reported in the literature such as mixed-integer linear programming (MILP) approaches \cite{lomuscio2017an_arxiv,dutta2017output}, linear programming (LP) based approaches \cite{ehlers2017formal}, the Reluplex algorithm that stems from the Simplex algorithm \cite{katz2017reluplex}, and polytope-operation-based approaches \cite{xiang2017reachable_arxiv,tran2019star}. For neural networks with general activation functions, the sensitivity for neural networks was introduced in \cite{zeng2001sensitivity,zeng2003quantified} and used for various problems, for instance, 
weight selection \cite{piche1995selection}, learning algorithm improvement \cite{xi2013architecture}, and architecture construction \cite{shi2005sensitivity}. Based on the maximal sensitivity concept,  a simulation-based verification approach is introduced in \cite{xiang2018output}. The output reachable set estimation for feedforward neural networks with general activation functions is formulated in terms of four convex optimization problems. These results are able to compute estimated and exact output sets of a feedforward neural network, and it, therefore, implies the availability of reachable set estimation and safety verification of closed-loop systems equipped with neural network controllers as shown in \cite{xiang2018reachable_acc,xiang2019reachability,dutta2019reachability}. The Verisig approach \cite{ivanov2019verisig},  transforms a neural network controller with sigmoid activation functions to an equivalent nonlinear hybrid system. This is combined with plant dynamics by using ODE reachability analysis routines for safety verification. All those existing methods were developed mainly based on exploiting the neural network itself such as the piecewise linear feature of ReLU activation functions or transformation of neural networks. In this work, we emphasize that our method focuses on both interval-based derivations of neural networks as well as taking advantage of simulations originated from neural networks for reachable set computation and safety verification of neural network control systems. 

\subsection{Contributions}
This paper focuses on improving the simulation-based approach developed in  \cite{xiang2018output} for the output set over-approximation of feedforward neural networks with general activation functions. A novel adaptive simulation-guided method will be developed and further integrated for safety verification of closed-loop systems with neural network controllers. 
In this paper, we develop a novel simulation-guided approach to perform the output reachable set estimation of feedforward neural networks with general activation functions. The algorithm is formulated in the framework of interval arithmetic and under the guidance of a finite number of simulations. The developed method using the information of simulations is able to provide much less computational cost than the previous paper \cite{xiang2018output}. As shown by a robotic arm model example, it only needs about 3\% computational cost of the method proposed in \cite{xiang2018output} to obtain a same interval-based reachability analysis result. We also extend our reachable set estimation result for safety verification of neural network control systems, in which  plants are modeled by ordinary differential equations (ODEs). We develop an algorithm to compute the reachable set of a neural network control system modeled by sampled-data systems. Based on the reachable set estimation, a safety verification algorithm is developed to provide formal safe assurance for neural network control systems, and an adaptive cruise control system using a software prototype is proposed to demonstrate our method.

\section{System Description and Problem Formulation}
\subsection{Neural Network Control Systems}
In this paper we consider a class of continuous-time nonlinear systems in the form of
\begin{align} \label{system}
\left\{ {\begin{array}{*{20}l}
	\dot{\mathbf{x}}(t) = f(\mathbf{x}(t),\mathbf{u}(t))\\
	\mathbf{y}(t) = h(\mathbf{x}(t))
	\end{array} } \right.
\end{align}
where $\mathbf{x}(t) \in \mathbb{R}^{n_x}$ is the state vector, $\mathbf{u}(t) \in \mathbb{R}^{n_u}$  is the control input and $\mathbf{y}(t) \in \mathbb{R}^{n_y}$ is the  controlled output, respectively. The  nonlinear controller is considered in its general form of 
\begin{align}
\mathbf{u}(t) = \gamma(\mathbf{y}(t),\mathbf{v}(t),t)
\end{align}
where $\mathbf{v}(t) \in \mathbb{R}^{n_v}$ is the reference input. As we know, the controller design problem for nonlinear systems in the general form is quite challenging and still open even $f$ and $h$ are available. To avoid the difficulties arising in such controller design problems for systems with complex model or even model unavailable, some data-driven approaches which only rely on the input-output data of the system were developed. In this paper, we consider a class of feedforward neural network trained by input-output data as the feedback controller of  dynamical systems. The feedforward neural network is in the following general form of
\begin{align}\label{neural_network_controller}
\mathbf{u}(t) = \Phi(\mathbf{y}(t),\mathbf{v}(t))
\end{align}
where $\Phi: \mathbb{R}^{n_y} \times \mathbb{R}^{n_v} \to \mathbb{R}^{n_u}$ is a neural network trained by data collected during system operations. We can rewrite the neural network controller in a more compact form of 
\begin{align}
\label{compact_neural_network_controller}
\mathbf{u}(t) = \Phi(\boldsymbol{\upeta}(t))
\end{align}
where $\boldsymbol{\upeta}(t) = [\mathbf{y}^{\top}(t)~\mathbf{v}^{\top}(t)]^{\top}$.

In practice, it always takes certain amount of time to compute the output signal of the neural network as the control input of the controlled plant. Hence, the neural network controller produces the control signals at every sampling time instant $t_k$, $k \in \mathbb{N}$, and then the controller maintains its value between two successive sampling instants $t_k$ and $t_{k+1}$. Due to the sampling mechanism of practical control systems, we can formulate the sampled neural network controller in the form of
\begin{align}\label{con_neural_network_controller}
\mathbf{u}(t) = \Phi(\boldsymbol{\upeta}(t_k)),~t \in [t_k,t_{k+1}).
\end{align}
and by substituting the above controller into  system (\ref{system}), we can obtain the closed-loop system in the following form
\begin{align}\label{sam_closed_loop}
\left\{ {\begin{array}{*{20}l}
	\dot{\mathbf{x}}(t) = f(\mathbf{x}(t),\Phi(\boldsymbol{\upeta}(t_k)))\\
	\mathbf{y}(t) = h(\mathbf{x}(t))
	\end{array} } \right.,~t \in [t_k,t_{k+1})
\end{align}
where $\boldsymbol{\upeta}(t_k) = [\mathbf{y}^{\top}(t_k)~ \mathbf{v}^{\top}(t_k)]^{\top}$. The mechanism of a sampled-data neural network system in the form of (\ref{sam_closed_loop}) is illustrated in Fig. \ref{nncs}. It worth mentioning that sampled-data model for neural network control systems in the form of (\ref{sam_closed_loop}) can be found in  a variety of articles such as \cite{wu2014exponential}. 

\begin{figure}
	\includegraphics[width=8.5cm]{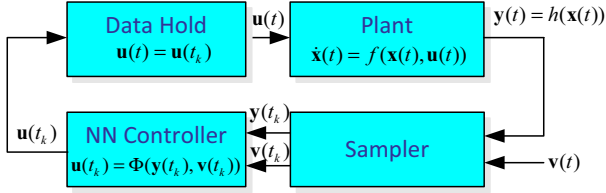}
	\caption{\boldmath Block diagram for continuous-time sampled-data neural network control systems.  }
	\label{nncs} 
\end{figure}

For the reachable set estimation of a neural network control system in the form of (\ref{sam_closed_loop}), the main challenge comes from the reachable set computation of the neural network controller since there are extensive reachable set estimation tools available for continuous-time ODE models. Even though a feedforward neural network is able to approximate any real-value function according to the Universal Approximation Theorem \cite{hornik1989multilayer},   a well-trained neural network will generate completely incorrect output if even slight disturbances are imposed to their input. Thus, the safety would become a major issue if we deploy neural network controllers to safety-critical control systems.

According to the Universal Approximation Theorem \cite{hornik1989multilayer}, it guarantees that, in principle, such a feedforward neural network, namely the function $\Phi(\cdot)$, is able to approximate any nonlinear real-valued function. Despite the impressive ability to approximate nonlinear functions,  much complexities represent in predicting the output behaviors of neural network controllers as well as the closed-loop systems. In most real applications, a neural network is usually viewed as a black box to generate a desirable output with respect to a given input. However, regarding property verification such as the safety verification, it has been observed that even a well-trained neural network can react in unexpected and incorrect ways to even slight perturbations of their inputs, which could result in unsafe even unstable systems. Thus, to verify safety properties of dynamical systems with neural network components, it is necessary to perform reachability analysis for the closed-loop system in the form of  (\ref{sam_closed_loop}) over a given finite time horizon, which is able to cover all possible values of system state in the given interval, to assure that the state trajectories of the closed-loop system will not attain unreasonable or even unsafe values. 

\subsection{Feedforward Neural Networks}

In this paper, we consider a class of feedforward neural networks which is called multilayer perceptrons (MLP). The basic processing elements in MLP are neurons which are defined by  activation functions in the form of 
A neural network consists of a number of interconnected neurons. Each neuron is a simple processing element that responds to the weighted inputs it received from other neurons. In this paper, we consider the most popular and general feedforward neural network, multilayer perceptrons (MLP). Generally, an MLP consists of three typical classes of layers: An input layer, that serves to pass the input vector to the network, hidden layers of computation neurons, and an output layer composed of at least a computation neuron to generate the output vector. The action of a neuron depends on its activation function, which is described as 
\begin{align}
u_i = \phi\left(\sum\nolimits_{j=1}^{n}\omega_{ij} \eta_j + \theta_i\right)
\end{align}
where $\eta_j$ is the $j$th input of the $i$th neuron, $\omega_{ij}$ is the weight from the $j$th input to the $i$th neuron, $\theta_i$ is the bias of the $i$th neuron, $u_i$ is the output of the $i$th neuron, and $\phi(\cdot)$ is the activation function. There are a variety of activation functions such as ReLU, tanh, logistic. In this paper, our approach is able to deal with activation functions without considering their specific forms. 

In this work, we assume the MLP has $L$ layers, and each layer $\ell$, $1 \le \ell \le L$ consists of $n^{\{\ell\}}$ neurons. Especially, we use layer $\ell =0$ to denote the input layer where  the input vector is passed to the network, and $n^{\{0\}}$ is number of the inputs for the network. In addition, $n^{\{L\}}$ is used to denote the output layer. For the layer $\ell$, the input vector of the layer $\ell$ is $\boldsymbol{\upeta}^{\{\ell\}}$, and the weight matrix and the bias vector are
\begin{align}
\mathbf{W}^{\{\ell\}} = \left[\omega_{1}^{\{\ell\}},\ldots,\omega_{n^{\{\ell\}}}^{\{\ell\}}\right]^{\top}
\\
\boldsymbol {\uptheta}^{\{\ell\}}=\left[\theta_1^{\{\ell\}},\ldots,\theta_{n^{\{\ell\}}}^{\{\ell\}}\right]^{\top}
\end{align} 
and the output vector of layer $\ell$ can be written in the form of 
\begin{equation}
\mathbf{u}^{\{\ell\}}=\phi_{\ell}(\mathbf{W}^{\{\ell\}}\boldsymbol{\upeta}^{\{\ell\}}+\uptheta^{\{\ell\}})
\end{equation} 
where $\phi_{\ell}(\cdot)$ denotes the activation function of layer $\ell$.

As the output of  layer $\ell$ equals the input of its successive layer $\ell+1$,  we can obtain the mapping from the input vector of input layer $\ell=0$ to output vector of the output layer $\ell =L$. Namely, the input-output relationship of an MLP can be expressed in the following form
\begin{equation}\label{NN}
\mathbf{u}^{\{L\}} = \Phi (\boldsymbol{\upeta}^{\{0\}})
\end{equation}    
where $\Phi(\cdot) \triangleq \phi_L  \circ \phi_{L - 1}  \circ  \cdots  \circ \phi_1(\cdot) $.

\subsection{Problem Formulation}
Given an input set, the output set of an MLP is given by the following definition.
\begin{definition}
	Given an input set $\mathcal{H}$, the output set of the MLP in the form of (\ref{NN}) is 
	\begin{align}
	\mathcal{U} = \left\{\mathbf{u} ^{\{L\}} \in \mathbb{R}^{n_u} \mid \mathbf{u}^{\{L\}} = \Phi (\boldsymbol{\upeta}^{\{0\}}),~ \boldsymbol{\upeta}^{\{0\}} \in \mathcal{H}\right\}.
	\end{align}
\end{definition}

The exact output set of an MLP is extremely difficult to obtain due to the complexity of neural networks. We often resort to compute an over-approximation of $\mathcal{U}$ which would be more feasible and practical.

\begin{definition}
	Given the output set $\mathcal{U}$ of MLP (\ref{NN}), if there exist a set $\mathcal{U}_e$ such as $\mathcal{U}\subseteq
	\mathcal{U}_e$ holds, then $\mathcal{U}_e$ is an output reachable set estimation of MLP (\ref{NN}).
\end{definition}

The first key issue that needs to be addressed in this paper is the reachable set estimation for MLP in the form of (\ref{NN}), which is summarized as follows.

\begin{problem}\label{problem1}
	Given an MLP in the form of (\ref{NN}) and a bounded set $\mathcal{H}$ as input set,  how does one compute a set $\mathcal{U}_e$ such that $\mathcal{U} \subseteq \mathcal{U}_e$ holds and moreover,
	the set $\mathcal{U}_e$ is required to be as small as possible\footnote{For a set $\mathcal{U}$, its over-approximation $\mathcal{U}_{e,1}$ is smaller than another over-approximation ${\mathcal{U}}_{e,2}$ if $d_{\mathrm{H}}({\mathcal{U}}_{e,1},\mathcal{U}) < d_{\mathrm{H}}({\mathcal{U}}_{e,2},\mathcal{U})$ holds, where $d_H$ stands for the Hausdorff distance.}?
\end{problem} 

Then, let us consider the neural network control system (\ref{sam_closed_loop}). The state trajectory of the closed-loop system (\ref{sam_closed_loop}) from a single initial state $\mathbf{x}_0$ is denoted by $\mathbf{x}(t; \mathbf{x}_0,\mathbf{v}(\cdot))$, where $t \in \mathbb{R}_{\ge 0}$ is the time and $\mathbf{v}(\cdot)$ stands for the input trajectory. With an initial set and input set, the reachable set for the closed-loop system (\ref{sam_closed_loop}) is given as follows.
\begin{definition}
	Given a neural network control system in the form of (\ref{sam_closed_loop}), an initial set $\mathcal{X}_0$ and an input set $\mathcal{V}$, the reachable set at  time instant $t$ is defined by
	\begin{equation}
	\mathcal{R}(t) =\left\{\mathbf{x}(t;\mathbf{x}_0,\mathbf{v}(\cdot))\in\mathbb{R}^{n_x} \mid \mathbf{x}_0 \in \mathcal{X}_0, \mathbf{v}(t) \in \mathcal{V}\right\}
	\end{equation}
	and the reachable set of system (\ref{sam_closed_loop}) over time interval $[t_0,t_f]$ is defined by
	\begin{equation}
	\mathcal{R}([t_0,t_f]) = \bigcup\nolimits_{t \in [t_0,t_f]}\mathcal{R}(t) .
	\end{equation}
\end{definition}

Similarly, for most of the system classes, the exact reachable set cannot be computed. Instead, we resort to derive over-approximations for the purpose of safety verification. 

\begin{definition}\label{def2}
	Given system (\ref{sam_closed_loop}) and its reachable set  $\mathcal{R}(t)$, $\mathcal{R}_e(t)$ is an over-approximation of $\mathcal{R}(t)$
	at time $t$ if $\mathcal{R}(t) \subseteq \mathcal{R}_e(t)$ holds. Moreover, 
	$\mathcal{R}_e([t_0,t_f]) = \bigcup\nolimits_{t \in[t_0,t_f]} \mathcal{R}_e(t)$ is an
	over-approximation of $\mathcal{R}([t_0,t_f])$ over interval $[t_0, t_f]$.
\end{definition}

The main problem, the reachable set estimation problem for neural network control system (\ref{sam_closed_loop}), is 
summarized as below.
\begin{problem}\label{problem2}
	Given closed-loop system (\ref{sam_closed_loop}), a bounded initial set $\mathcal{X}_0$ and an input set $\mathcal{V}$, how does one find a set $\mathcal{R}_e(t)$ such that $\mathcal{R}(t) \subseteq \mathcal{R}_e(t)$ holds?
\end{problem}

Based on the reachable set estimation of neural network control systems, the safety verification for such dynamical systems can be performed. The safety specification is defined by a set the state space, which describes the safety requirement for the system.
\begin{definition}
	Given neural network control system (\ref{sam_closed_loop}) and a safety specification set $\mathcal{S}$ which formalizes the safety requirements. The closed-loop system (\ref{sam_closed_loop}) is safe during time interval $[t_0, t_f ]$ if and only if the following condition holds:
	\begin{equation}\label{verification}
	\mathcal{R}([t_0,t_f])  \cap \neg \mathcal{S} = \emptyset
	\end{equation}
	where $\neg$ is the logical negation symbol.
\end{definition}

Therefore, the safety verification problem for neural network control system (\ref{sam_closed_loop}) is as follows.
\begin{problem}\label{problem3}
	Given a neural network control system in the form of (\ref{sam_closed_loop}), a bounded initial set $\mathcal{X}_0$, an input set $\mathcal{V}$ and a
	safety specification set $\mathcal{S}$, how does one examine if the safety requirement  (\ref{verification}) holds?
\end{problem}

Before ending this section, a useful lemma is presented, which implies that the safety verification of neural network control system (\ref{sam_closed_loop}) can be relaxed with the help of the reachable set estimation $\mathcal{R}_e$, instead of using the exact reachable set $\mathcal{R}$. 

\begin{lemma}\label{lemma1}
	Given a neural network control system in the form of (\ref{sam_closed_loop}) and a safety specification set $\mathcal{S}$, the closed-loop system (\ref{sam_closed_loop}) is safe over time interval $[t_0,t_f]$ if the following condition holds
	\begin{equation}\label{lemma1_1}
	\mathcal{R}_e([t_0,t_f]) \cap \neg \mathcal{S} = \emptyset
	\end{equation}
	where $\mathcal{R}([t_0,t_f])  \subseteq\mathcal{R}_e([t_0,t_f])$.
\end{lemma}
\begin{proof}
	Because of $\mathcal{R}([t_0,t_f])  \subseteq\mathcal{R}_e([t_0,t_f])$, condition (\ref{lemma1_1}) implies $\mathcal{R}([t_0,t_f]) \cap\neg \mathcal{S} = \emptyset$.
	The proof is complete.
\end{proof}

Lemma \ref{lemma1} implies that the over-approximated reachable set $\mathcal{R}_e([t_0,t_f])$ is qualified for the safety verification over interval $[t_0,t_f]$. The three linked problems are the main concerns to be addressed in the rest of the paper. Essentially, the very first and basic problem is the Problem \ref{problem1}, namely finding efficient methods to estimate the output set of an MLP. In the remainder of this paper, attention is mainly devoted to give solutions for Problem \ref{problem1}, and then extend to solve Problems \ref{problem2} and \ref{problem3}.

\section{Simulation-Guided Reachability Analysis for Neural Networks}
In this section, we will first present a thorough interval analysis for MLPs, and then propose our key contribution, the simulation-guided reachable set estimation algorithm. In the end, a robotic arm model example is proposed to demonstrate our approach. 
\subsection{Preliminaries}
Let $[x] = [\underline{x}, \overline{x}]$, $[y] = [\underline{y},\overline{y}]$ be real compact intervals and $\circ$ denote one of the basic operations including addition,
subtraction, multiplication and division, respectively, for real numbers, that is $\circ \in \{+,-,\cdot, / \}$, where it is assumed that $0 \notin [b]$ in case of division. We define these operations for intervals $[x]$ and $[y]$ by $[x] \circ [y] = \{x \circ y \mid x \in [y],x\in [y]\}$. The width of an interval $[x]$ is defined and denoted by $w([x]) = \overline{x} - \underline{x}$. The set of compact intervals in $\mathbb{R}$ is denoted by $\mathbb{IR}$. We say $[\phi]: \mathbb{IR} \to \mathbb{IR}$ is an interval extension of function $\phi: \mathbb{R} \to \mathbb{R}$, if for any degenerate interval arguments, $[\phi]$ agrees with $\phi$ such that $[\phi]([x,x]) = \phi(x)$. In order to consider multidimensional problems where $\mathbf{x} \in \mathbb{R}^{n}$ is taken into account, we denote $[\mathbf{x}] =[\underline{x}_1,\overline{x}_1]\times\cdots \times[\underline{x}_n,\overline{x}_n] \in \mathbb{IR}^{n}$, where $\mathbb{IR}^n$ denotes the set of compact interval in $\mathbb{R}^n$. The width of an interval vector $\mathbf{x}$ is the largest of the widths of any
of its component intervals $w([\mathbf{x}])= \max_{i=1,\ldots,n} (\overline{x}_i-\underline{x}_i)$. A mapping $[\Phi] : \mathbb{IR}^{n} \to \mathbb{IR}^{m}$ denotes the interval extension of a  function $\Phi:\mathbb{R}^{n} \to  \mathbb{R}^m$. An interval extension is inclusion monotonic if, for any $[\mathbf{x}_1],[\mathbf{x}_2] \in \mathbb{IR}^{n}$, $[\mathbf{x}_1] \subseteq [\mathbf{x}_2]$ implies $[\Phi]([\mathbf{x}_1]) \subseteq [\Phi]([\mathbf{x}_2])$. A fundamental property of inclusion monotonic interval extensions is that $\mathbf{x} \in [\mathbf{x}] \Rightarrow \Phi(\mathbf{x}) \in [\Phi]([\mathbf{x}])$, which means the value of $\Phi$ is contained in the interval $[\Phi]([\mathbf{x}])$ for every $\mathbf{x}$ in $[\mathbf{x}]$.

Several useful definitions and lemmas are presented.

\begin{definition} \emph{\cite{moore2009introduction}}
	Piecewise monotone functions, including absolute value, exponential, logarithm, rational power,   and trigonometric functions, are standard functions.
\end{definition}

\begin{lemma} \label{lemma2}\emph{\cite{moore2009introduction}}
	A function $\Phi$ composed by finitely many standard functions with elementary operations $\{+,-,\cdot, / \}$ is inclusion monotone.
\end{lemma}

\begin{definition} \emph{\cite{moore2009introduction}}
	Given an interval extension $[\Phi]([\mathbf{x}])$, if there is a constant $\xi$
	such that $w([\Phi]([\mathbf{x}]))\le \xi w([\mathbf{x}])$ for every $[\mathbf{x}] \subseteq [\mathbf{x}_0]$, then$[\Phi]([\mathbf{x}])$ is said to be Lipschitz in $[\mathbf{x}_0]$ .
\end{definition}

\begin{lemma}\label{lemma3}\emph{\cite{moore2009introduction}}
	If a function $\Phi(\mathbf{x})$ satisfies a Lipschitz condition in $[\mathbf{x_0}]$,
	\begin{equation}
	\left\|\Phi(\mathbf{x}_2)-\Phi(\mathbf{x}_1)\right\| \le \xi\left\|\mathbf{x}_2-\mathbf{x}_1\right\|,~\mathbf{x}_1,\mathbf{x}_2 \in [\mathbf{x}_0]
	\end{equation}
	then the interval extension $[\Phi]([\mathbf{x}])$ is a Lipschitz interval extension in $[\mathbf{x}_0]$, 
	\begin{equation}
	w([\Phi]([\mathbf{x}]))\le \xi w([\mathbf{x}]),~[\mathbf{x}] \subseteq [\mathbf{x_0}].
	\end{equation}
\end{lemma}


\begin{assumption}\label{assumption_0}
	The activation function $\phi$ considered in this paper is composed by standard functions with finitely many elementary operations. 
\end{assumption}

The above assumption allows that the reachability analysis of MLP can be conducted in the framework of interval arithmetic, and to our knowledge, popular activation functions such as tanh, sigmoid, ReLU satisfy this assumption.  

\subsection{Interval Analysis}

First, we consider a single layer $\mathbf{u} = \phi(\mathbf{W}\boldsymbol{\upeta}+\boldsymbol{\uptheta})
$. Given an interval input $[\boldsymbol{\upeta}]$, the interval extension is $[\phi](\mathbf{W}[\boldsymbol{\upeta}]+\boldsymbol{\uptheta}) = [\underline{u}_1,\overline{u}_1]\times\cdots\times[\underline{u}_n,\overline{u}_n] = [\mathbf{u}]$, where 
\begin{align}
\underline{u}_i &= \min_{\boldsymbol{\upeta} \in [\boldsymbol{\upeta}]} \phi\left(\sum\nolimits_{j=1}^{n}\omega_{ij} \eta_j + \theta_i\right)\label{thm1_1}
\\
\overline u_i &= \max_{\boldsymbol{\upeta} \in [\boldsymbol{\upeta}]} \phi\left(\sum\nolimits_{j=1}^{n}\omega_{ij} \eta_j + \theta_i\right) . \label{thm1_2}
\end{align} 

According to (\ref{thm1_1}) and (\ref{thm1_2}), the minimum and maximum values of the output of nonlinear function $\phi$ is required to compute the interval extension $[\phi]$. However, the optimization problems are still challenging for general nonlinear functions. We propose the following monotonic assumption for activation functions.

\begin{assumption}\label{assumption_1}
	For any two scalars $\eta_1 \le \eta_2$, the activation function satisfies $\phi(\eta_1) \le \phi(\eta_2)$. 
\end{assumption}

It worth mentioning that Assumption \ref{assumption_1} can be satisfied by a variety of activation functions such as logistic, tanh,  ReLU, all satisfy Assumption \ref{assumption_1}. Taking advantage of the monotonic property of $\phi$, we have  interval extension $[\phi]([\eta]) = [\phi(\underline{\eta}),\phi(\overline{\eta})]$. Therefore, $\underline{u}_i$ and $\overline{u}_i$ in (\ref{thm1_1}) and (\ref{thm1_2}) can be explicitly written out  as
\begin{align} \label{y_1}
\underline{u}_i & = \phi\left( \sum\nolimits_{j=1}^{n}\underline{p}_{ij} +   \theta_i\right)
\\
\overline{u}_i &= \phi\left( \sum\nolimits_{j=1}^{n}\overline{p}_{ij} +   \theta_i \right) \label{y_2}
\end{align}
with $\underline{p}_{ij}$ and $\overline{p}_{ij}$ defined by
\begin{align} \label{y_3}
\underline{p}_{ij}  &= \left\{ {\begin{array}{*{20}l}
	{\omega _{ij} \underline{\eta}_j,} & {\omega _{ij}\geq 0}  \\
	{\omega _{ij} \overline \eta_j ,} & {\omega _{ij}  < 0}  \\
	\end{array} } \right.
\\
\overline p_{ij}& = \left\{ {\begin{array}{*{20}c}
	{\omega _{ij} \overline \eta_j ,} & {\omega _{ij}  \geq 0}  \\
	{\omega _{ij} \underline{\eta}_j ,} & {\omega _{ij}  < 0}  \\
	\end{array} } \right.. \label{y_4}
\end{align}

From (\ref{y_1})--(\ref{y_4}), the output interval of a single layer can be efficiently computed with these explicit expressions. Then, we consider the MLP $\mathbf{u}^{\{L\}}=\Phi(\boldsymbol{\upeta}^{\{0\}})$ with multiple layers, the interval extension $[\Phi ]([\boldsymbol{\upeta}^{\{ 0\} } ])$ can be computed by the following layer-by-layer computation. 
\begin{theorem}\label{thm2}
	Given an MLP in the form of (\ref{NN}) with activation functions satisfying Assumption \ref{assumption_1} and an interval input $[\boldsymbol{\upeta}^{\{0\}}]$, an interval extension can be determined by
	\begin{equation} \label{thm2_1}
	[\Phi ]([\boldsymbol{\upeta}^{\{ 0\} } ]) = [\hat \phi _L ] \circ  \cdots  \circ [\hat \phi _1 ] \circ [\hat \phi _0 ]([\boldsymbol{\upeta}^{\{ 0\} } ])
	\end{equation}
	where $[\hat \phi_{\ell}]([\boldsymbol{\upeta}^{\{\ell\}}]) =[\phi_{\ell} ](\mathbf{W}^{\{\ell\}} [\boldsymbol{\upeta}^{\{\ell\}} ] + \boldsymbol{\uptheta}^{\{\ell\}}  )=[\mathbf{u}^{\{\ell\}}]$ in which
	\begin{align} \label{thm2_2}
	\underline{u}_i^{\{\ell\}} & = \phi_{\ell}\left( \sum\nolimits_{j=1}^{n^{\{\ell\}}}\underline{p}_{ij}^{\{\ell\}} +   \theta_i^{\{\ell\}} \right)
	\\
	\overline{u}_i^{\{\ell\}} &= \phi_{\ell}\left(\sum\nolimits_{j=1}^{n^{\{\ell\}}}\overline{p}_{ij}^{\{\ell\}} +   \theta_i^{\{\ell\}} \right) \label{thm2_3}
	\end{align}
	with $\underline{p}_{ij}^{\{\ell\}}$ and $\overline{p}_{ij}^{\{\ell\}}$ defined by
	\begin{align} \label{thm2_4}
	\underline{p}_{ij}^{\{\ell\}}  &= \left\{ {\begin{array}{*{20}l}
		{\omega _{ij}^{\{\ell\}} \underline{\eta}_j^{\{\ell\}},} & {\omega _{ij}^{\{\ell\}}\geq 0}  \\
		{\omega _{ij}^{\{\ell\}} \overline \eta_j^{\{\ell\}} ,} & {\omega _{ij}^{\{\ell\}}  < 0}  \\
		\end{array} } \right.
	\\
	\overline p_{ij}^{\{\ell\}}& = \left\{ {\begin{array}{*{20}c}
		{\omega _{ij}^{\{\ell\}} \overline \eta_j^{\{\ell\}} ,} & {\omega _{ij}^{\{\ell\}}  \geq 0}  \\
		{\omega _{ij}^{\{\ell\}} \underline{\eta}_j^{\{\ell\}} ,} & {\omega _{ij}^{\{\ell\}}  < 0}  \\
		\end{array} } \right.. \label{thm2_5}
	\end{align}
\end{theorem}
\begin{proof}
	We denote $\hat\phi_{\ell}(\boldsymbol{\upeta}^{\{\ell\}}) = \phi_{\ell} (\mathbf{W}^{\{\ell\}} \boldsymbol{\upeta}^{\{\ell\}}  + \boldsymbol{\uptheta}^{\{\ell\}}  )$. Given an MLP, it essentially has $\boldsymbol{\upeta}^{\{\ell\}}=\hat\phi_{\ell-1}(\boldsymbol{\upeta}^{\{\ell-1\}})$, $\ell=1,\ldots,L$ which leads to (\ref{thm2_1}). Then, for each layer, the interval extension $[\mathbf{u}^{\{\ell\}}]$ computed by (\ref{thm2_2})--(\ref{thm2_5}) can be obtained by (\ref{y_1})--(\ref{y_4}). 
\end{proof}

According to the explicit expressions (\ref{thm2_1})--(\ref{thm2_5}), the computation on interval extension  $[\Phi]$ can be performed in a fast manner. In the next step, we should discuss the conservativeness for the computation outcome of (\ref{thm2_1})--(\ref{thm2_5}). 

\begin{theorem}\label{thm1}
	The interval extension $[\Phi ]$ of neural network $\Phi$ composed by activation functions satisfying Assumption \ref{assumption_0} is inclusion monotonic and Lipschitz such that 
	\begin{equation}\label{L_NN}
	w([\Phi]([\boldsymbol{\upeta}]))\le \xi^{L}\prod\nolimits_{\ell = 1}^L {\left\| {\mathbf{W}^{\{ \ell\} } } \right\|}  w([\boldsymbol{\upeta}]),~[\boldsymbol{\upeta}] \subseteq \mathbb{IR}^{n^{\{0\}}}
	\end{equation}
	where $\xi$ is a Lipschitz constant for  activation functions in $\Phi$. 
\end{theorem}
\begin{proof}
	Under Assumption \ref{assumption_0}, the inclusion monotonicity can be obtained directly based on Lemma \ref{lemma2}. Then,  we denote $\hat\phi_{\ell}(\boldsymbol{\upeta}^{\{\ell\}}) = \phi_{\ell} (\mathbf{W}^{\{\ell\}}\boldsymbol{\upeta}^{\{\ell\}}  + \boldsymbol{\uptheta}^{\{\ell\}}  )$. For any $\boldsymbol{\upeta}_1,\boldsymbol{\upeta}_2$, it has
	\begin{align*}
	\left\| {\hat \phi _{\ell} (\boldsymbol{\upeta}_2^{\{ \ell\} } ) - \hat \phi _{\ell} (\boldsymbol{\upeta}_1^{\{ \ell\} } )} \right\| \leq \xi \left\| {\mathbf{W}^{\{ \ell\} } \boldsymbol{\upeta}_2^{\{ \ell\} }  - \mathbf{W}\boldsymbol{\upeta}_1^{\{ \ell\} } } \right\| \nonumber
	\\
	\leq \xi \left\| {\mathbf{W}^{\{ \ell\} } } \right\|\left\| {\boldsymbol{\upeta}_2^{\{ \ell\} }  - \boldsymbol{\upeta}_1^{\{ \ell\} } } \right\|.
	\end{align*}
	
	Due to $\boldsymbol{\upeta}^{\{\ell\}}=\hat\phi_{\ell-1}(\boldsymbol{\upeta}^{\{\ell-1\}})$, $\ell=1,\ldots,L$,  $\xi^{L}\prod\nolimits_{\ell = 1}^L {\left\| {\mathbf{W}^{\{ \ell\} } } \right\|}$ becomes the Lipschitz constant for $\Phi$, and (\ref{L_NN}) can be established by Lemma \ref{lemma3}. 
\end{proof}

We denote the set image for neural network $\Phi$ as follows
\begin{equation}
\Phi([\boldsymbol{\upeta}^{\{0\}}])=\{\Phi(\boldsymbol{\upeta}^{\{0\}}):\boldsymbol{\upeta}^{\{0\}} \in [\boldsymbol{\upeta}^{\{0\}}]\}.
\end{equation}

Since $[\Phi]$ is inclusion monotonic according to Theorem \ref{thm1}, one has $\Phi([\boldsymbol{\upeta}^{\{0\}}]) \subseteq [\Phi]([\boldsymbol{\upeta}^{\{0\}}])$.   We
have $[\Phi]([\boldsymbol{\upeta}^{\{0\}}]) = \Phi([\boldsymbol{\upeta}^{\{0\}}]) + E([\boldsymbol{\upeta}^{\{0\}}])$ for some interval-valued function $E([\boldsymbol{\upeta}^{\{0\}}])$ such that $w([\Phi]([\boldsymbol{\upeta}^{\{0\}}])) = w(\Phi([\boldsymbol{\upeta}^{\{0\}}])) + w(E([\boldsymbol{\upeta}^{\{0\}}]))$. 

\begin{definition}
	$w(E([\boldsymbol{\upeta}^{\{0\}}])) = w([\Phi]([\boldsymbol{\upeta}^{\{0\}}])) - w(\Phi([\boldsymbol{\upeta}^{\{0\}}])) $
	is the excess width of interval extension of neural network $\Phi([\boldsymbol{\upeta}^{\{0\}}])$.
\end{definition}

Explicitly, the excess width measures the conservativeness of interval extension $[\Phi]$ regarding its corresponding function $\Phi$. The following theorem gives the upper bound of the excess width $w(E([\boldsymbol{\upeta}^{\{0\}}]))$. 

\begin{theorem}\label{thm3}
	Given an MLP in the form of (\ref{NN}) with an interval input $[\boldsymbol{\upeta}^{\{0\}}]$, the excess width $w(E([\boldsymbol{\upeta}^{\{0\}}]))$ satisfies
	\begin{equation}\label{thm3_1}
	w(E([\boldsymbol{\upeta}^{\{ 0\} } ])) \leq\gamma w([\boldsymbol{\upeta}^{\{0\}} ])
	\end{equation}
	where $\gamma = \xi^{L}\prod\nolimits_{\ell = 1}^L {\left\| {\mathbf{W}^{\{ \ell\} } } \right\|} $.
\end{theorem}
\begin{proof}
	We have $[\Phi]([\boldsymbol{\upeta}^{\{0\}}]) = \Phi([\boldsymbol{\upeta}^{\{0\}}]) + E([\boldsymbol{\upeta}^{\{0\}}])$ for some $E([\boldsymbol{\upeta}^{\{0\}}])$ and 
	\begin{align*}
	w(E([\boldsymbol{\upeta}^{\{0\}}])) &= w([\Phi]([\boldsymbol{\upeta}^{\{0\}}])) - w(\Phi([\boldsymbol{\upeta}^{\{0\}}]))
	\\
	&\leq w([\Phi ]([\boldsymbol{\upeta}^{\{ 0\} } ])) 
	\\
	& \leq \xi ^L \prod\nolimits_{\ell = 1}^L {\left\| {\mathbf{W}^{\{ \ell\} } } \right\|} w([\boldsymbol{\upeta}^{\{0\}} ])
	\end{align*}
	which means (\ref{thm3_1}) holds.
\end{proof}

Given a neural network $\Phi$ which means $\mathbf{W}^{\{\ell\}}$ and $\xi$ are fixed, Theorem \ref{thm3} implies that a less conservative result can be only obtained by reducing the width of input interval $[\boldsymbol{\upeta}^{\{0\}}]$. On the other hand, a smaller $w([\boldsymbol{\upeta}^{\{0\}}])$ means more subdivisions of an input interval which will bring more computational cost. Therefore, how to generate appropriate subdivisions of an input interval is the key issue for reachability analysis of neural networks in the framework of interval arithmetic. In the next section, an efficient simulation-guided method is proposed to address this key problem.

\subsection{Simulation-Guided Reachability Analysis}

Inspired by the Moore-Skelboe algorithm \cite{skelboe1974computation}, we propose a reachable set computation algorithm under guidance of a finite number of simulations. It proposes an adaptive input interval partitioning scheme with the help of simulations. The simulation-guided algorithm shown in Algorithm \ref{alg1} checks the emptiness of the intersection between the computed output set and the over-approximation interval for simulations, within a pre-defined tolerance $\varepsilon$.  This algorithm is able to avoid unnecessary partition for the input interval to get a tight output range. The tightness of reachable set estimation is accomplished by dividing and checking the initial input interval into increasingly smaller sub-intervals, as seen in Algorithm \ref{alg1}.

\begin{itemize}
	\item \textbf{Initialization.}  Perform $N$ simulations for neural network $\Phi$ to get $N$ output points $\mathbf{u}_{\mathrm{sim},n}$, $n = 1,\ldots,N$ and compute an interval $[\mathbf{u}_{\mathrm{sim}}]$ such that $\mathbf{u}_{\mathrm{sim},n} \in [\mathbf{u}_{\mathrm{sim}}]$, $\forall n $. The $N$ simulations can be generated either randomly or by gridding input set. Since our approach is based on interval analysis, convert input set $\mathcal{H}$ to an interval $[\boldsymbol{\upeta}]$ such that $\mathcal{H} \subseteq [\boldsymbol{\upeta}]$. Compute the initial output interval $[\mathbf{u}] = [\Phi]([\boldsymbol{\upeta}])$ by (\ref{thm2_1})-(\ref{thm2_5}). Initialize set $\mathcal{M} = \{([\boldsymbol{\upeta}],[\mathbf{u}])\}$. Set a tolerance $\varepsilon>0$, which will be used to terminate algorithm.    
	\item \textbf{Simulation-guided bisection.}
	This is the key step in the algorithm. 
	Select an element $([\boldsymbol{\upeta}],[\mathbf{u}])$ for simulation-guided bisection. If the output interval $[\mathbf{u}]$ satisfies $[\mathbf{u}] \subseteq [\mathbf{u}_{\mathrm{sim}}]$, we can discard this sub-interval for the subsequent dividing and checking since it has been proven to be included in the output range. Otherwise, the bisection action will be activated to produce finer subdivisions to be added to $\mathcal{M}$ for subsequent checking. The bisection process is guided by simulations, since the bisection actions are totally determined by the non-emptiness of the intersection between output interval sets and the interval for simulations. This distinguishing feature leads to finer subdivisions when the output set is getting close to boundary of interval for simulations, and on the other hand coarse subdivisions are sufficient for interval reachability analysis when the output set is included in the interval for simulations. Therefore, unnecessary computational cost can be avoided. 
	\item \textbf{Termination.} The simulation-guided bisection  continues until the width of subdivisions becomes less than the pre-defined tolerance $\varepsilon$. Generally, a smaller tolerance $\varepsilon$ will lead to a tighter output interval computation result. 
\end{itemize}  

\begin{algorithm}[ht!]
	\SetAlgoLined
	\SetKwInOut{Input}{Input}
	\SetKwInOut{Output}{Output}
	\SetKw{Return}{return}
	\Input{Feedforward neural network $\Phi$
		; Input set $\mathcal{H}$; 
		Tolerance $\varepsilon$; Number of simulations $N$.}
	\Output{Output set estimation $\mathcal{U}_e$.}
	\Fn{\texttt{reachMLP}}{
		\tcc{Initialization}
		Compute interval $[\boldsymbol{\upeta}]$ such that $\mathcal{H} \subseteq [\boldsymbol{\upeta}]$ \;
		$[\mathbf{u}] \gets [\Phi]([\boldsymbol{\upeta}])$  \tcp*{Using (\ref{thm2_1})-(\ref{thm2_5})}
		$\mathcal{M} \gets \{([\boldsymbol{\upeta}],[\mathbf{u}])\}$ \;
		Compute $N$ simulations $\mathbf{u}_{\mathrm{sim},n} = \Phi(\boldsymbol{\upeta}_{\mathrm{sim},n})$, $n = 1,\ldots,N$ \; 
		Compute interval $[\mathbf{u}_{\mathrm{sim}}]$ such that $\mathbf{u}_{\mathrm{sim},n} \in [\mathbf{u}_{\mathrm{sim}}]$, $\forall n$ \;
		$\mathcal{U}_e \gets \emptyset$ \;
		\tcc{Simulation-guided bisection}
		\While{$\mathcal{M} \neq \emptyset$}{
			Select and remove an element $([\boldsymbol{\upeta}],[\mathbf{u}])$ from $\mathcal{M}$\;
			\eIf{$[\mathbf{u}]\subseteq  [\mathbf{u}_{\mathrm{sim}}]$}{
				$\mathcal{U}_e \gets \mathcal{U}_e \cup [\mathbf{u}]$ \;
				Continue \;
			}
			{\eIf{$w([\mathbf{\boldsymbol{\upeta}}]) > \varepsilon$}{Bisect $[\boldsymbol{\upeta}]$ to obtain $[\boldsymbol{\upeta}_1]$ and $[\boldsymbol{\upeta}_2]$ \;
					$[\mathbf{u}_1] \gets [\Phi]([\boldsymbol{\upeta}_1])$ \tcp*{Using (\ref{thm2_1})-(\ref{thm2_5})} $[\mathbf{u}_2] \gets [\Phi]([\boldsymbol{\upeta}_2])$ 
					\tcp*{Using (\ref{thm2_1})-(\ref{thm2_5})}
					$\mathcal{M} \gets \mathcal{M} \cup \{([\mathbf{u}_1],[\boldsymbol{\upeta}_1])\}\cup \{([\mathbf{u}_2],[\boldsymbol{\upeta}_2])\}$ \;}
				{Break \tcp*{Bisection terminates}}
			}
		}
		\Return{$\mathcal{U}_e \gets \mathcal{U}_e\cup\left(\bigcup_{\{[\boldsymbol{\upeta}],[\mathbf{u}]\} \in \mathcal{M}}[\mathbf{u}]\right)$}  }
	\caption{Simulation-Guided Reachable Set Estimation} \label{alg1}
\end{algorithm}
In the next subsection, we will use an illustrative example to show the advantages of the proposed simulation-guided approach. 

\subsection{Reachability Analysis of a Robotic Arm Model}
\begin{figure}
	\centering
	\includegraphics[width=3.5cm]{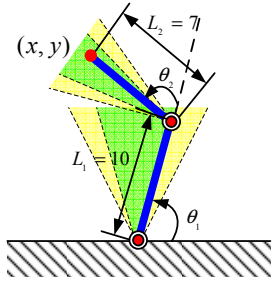}
	\caption{Robotic arm with two joints. The normal working zone of $(\theta_1,\theta_2)$ is colored in green $\theta_1,\theta_2 \in [\frac{5\pi}{12},\frac{7\pi}{12}]$. The buffering zone is in yellow $\theta_1,\theta_2 \in [\frac{\pi}{3},\frac{5\pi}{12}] \cup [\frac{7\pi}{12},\frac{2\pi}{3}] $. The forbidden zone is $\theta_1,\theta_2 \in [0,\frac{\pi}{3}] \cup [\frac{2\pi}{3},2\pi] $.}
	\label{robot_arm} 
\end{figure}
In \cite{xiang2018output}, a  \emph{learning forward kinematics} of a robotic arm model with two joints is proposed, shown in Fig. \ref{robot_arm}. 
The learning task is using a feedforward neural network to predict the position $(x,y)$ of the end with knowing the joint angles $(\theta_1,\theta_2)$. The input space $[0,2\pi]\times [0,2\pi]$ for $(\theta_1,\theta_2)$  is classified into three zones for its operations: Normal working zone $\theta_1,\theta_2 \in [\frac{5\pi}{12},\frac{7\pi}{12}]$, buffering zone $\theta_1,\theta_2 \in [\frac{\pi}{3},\frac{5\pi}{12}] \cup [\frac{7\pi}{12},\frac{2\pi}{3}] $ and forbidden zone $\theta_1,\theta_2 \in [0,\frac{\pi}{3}] \cup [\frac{2\pi}{3},2\pi]$. The detailed formulation for this robotic arm model and neural network training can be found in \cite{xiang2018output}. 

In \cite{xiang2018output}, a uniform partition of input interval which is the union of normal working  and buffering zones $(\theta_1,\theta_2) \in [\frac{\pi}{3},\frac{2\pi}{3}] \times [\frac{\pi}{3},\frac{2\pi}{3}]$, is used to compute an over-approximation for safety verification. The safety specification for the position $(x,y)$ is an interval set
$\mathcal{S}=\{(x,y)\mid -14 \le x\le 3~\mathrm{and}~1 \le y \le 17\}$. To illustrate the advantages of simulation-guided approach, we aim to compute a tight output interval using both uniform partition method in \cite{xiang2018output} and Algorithm \ref{alg1}. The precision/tolerance for both methods are chosen the same, $\varepsilon = 0.01$. The number of simulations used in Algorithm \ref{alg1} is set to be $1000$. The computed output ranges are shown in Figs. \ref{robot_arm_sim} and \ref{robot_arm_uniform}. It can be clearly observed that two methods can produce same output range estimations, that is $\mathcal{U}_e = \{(x,y)\mid -12.0258 \le x\le 1.1173~\mathrm{and}~2.8432 \le y \le 14.8902\}$ which is sufficient to ensure the safety due to $\mathcal{U}_e \subseteq\mathcal{S}$. Though both methods can achieve same output range analysis results, the computation costs are significantly different as shown in Table \ref{tab1}. In \cite{xiang2018output}, a uniform partition for input space is used, and it results in 16384 intervals with precision $\varepsilon = 0.01$ and the computation  takes $4.4254$ seconds. Using our simulation-guided approach, the safety can be guaranteed by partitioning the input space into  397 intervals ($2.42\%$ of those by uniform partition method in \cite{xiang2018output}) with tolerance $\varepsilon = 0.01$. The simulation-guided partition of the input interval $[\frac{\pi}{3},\frac{2\pi}{3}] \times [\frac{\pi}{3},\frac{2\pi}{3}]$ is shown in Fig. \ref{robot_arm_input}. Along with the less number of intervals involved in the computation process, the computational time is  0.1423 seconds ($3.22\%$ of that by uniform partition method in \cite{xiang2018output}) for simulation-guided approach\footnote{The  source code is available at:
	https://github.com/xiangweiming/ignnv}.  
\begin{table}
	\centering
	\caption{Comparison on number of intervals and computational time between simulation-guided method and uniform partitioning method. }\label{tab1}
	\begin{tabular}{c||c||c}
		\hline
		& Intervals 	& Computational Time   \\
		\hline
		Algorithm \ref{alg1}   &  $397$ &  $0.1423$ seconds \\
		\hline
		Xiang \textit{et al.} 2018 \cite{xiang2018output} & $16384$ & $4.4254$ seconds\\
		\hline
	\end{tabular}
\end{table} 
\begin{figure}[ht!]
	\includegraphics[width=8cm]{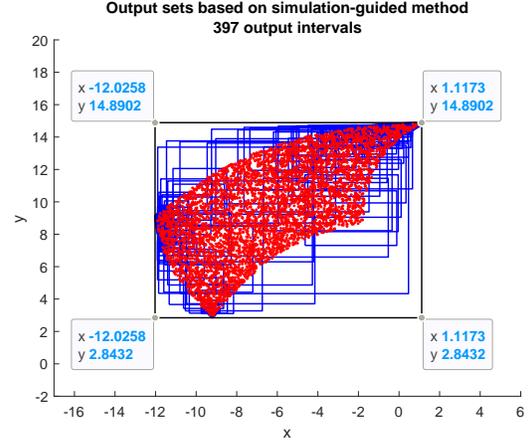}
	\caption{Output intervals obtained by simulation-guided methods. 397 output intervals (blue rectangles) are generated and the output range is $\mathcal{U}_e = \{(x,y)\mid -12.0258 \le x\le 1.1173~\mathrm{and}~2.8432 \le y \le 14.8902\}$ (black rectangle). Red points are 5000 random outputs which are all included in output intervals.}
	\label{robot_arm_sim} 
\end{figure}
\begin{figure}[ht!]
	\includegraphics[width=8cm]{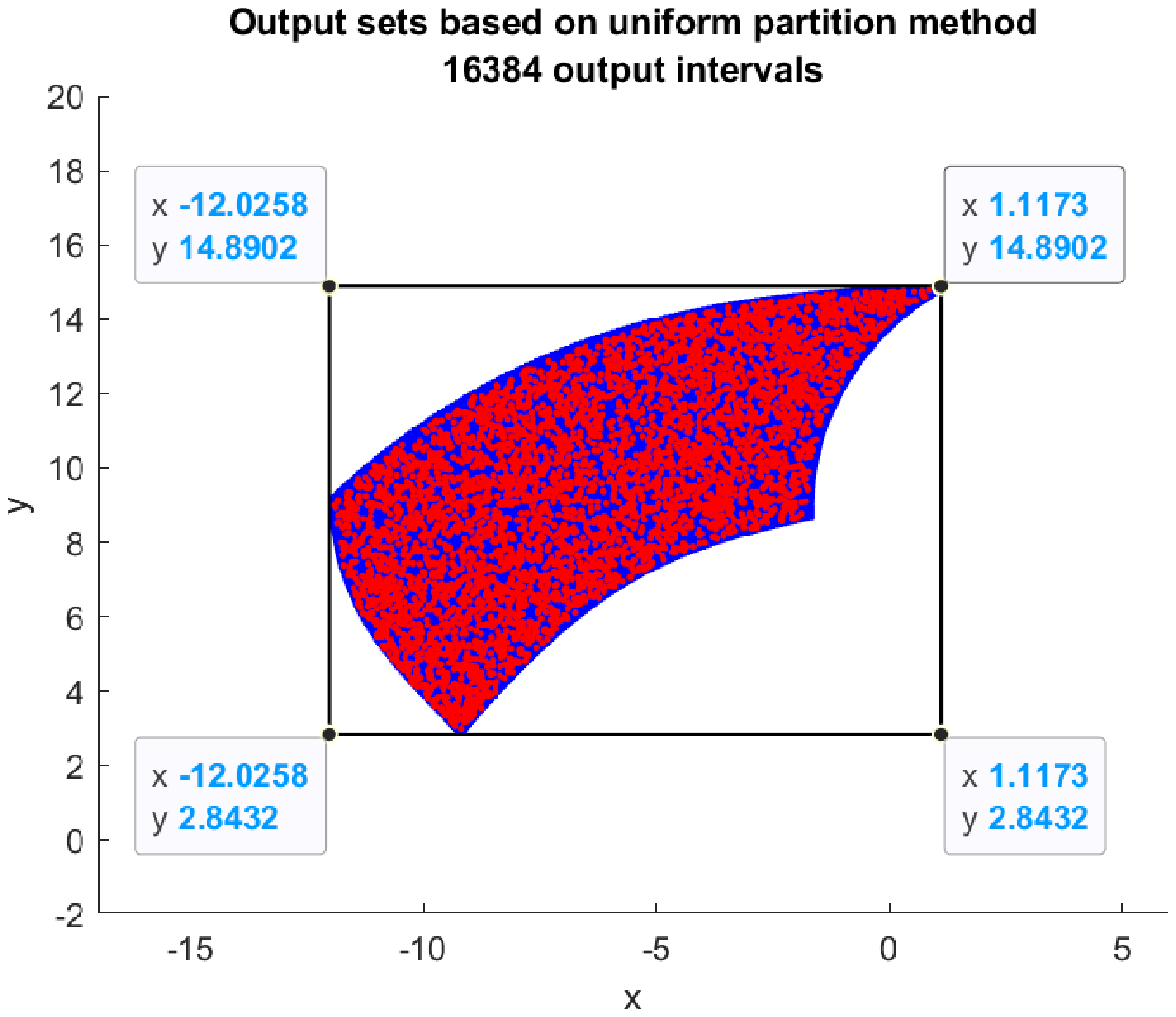}
	\caption{Output intervals obtained by uniform partition method in \cite{xiang2018output}. 16384 output intervals (blue rectangles) are generated and the output range is $\mathcal{U}_e = \{(x,y)\mid -12.0258 \le x\le 1.1173~\mathrm{and}~2.8432 \le y \le 14.8902\}$ (black rectangle). Red points are 5000 random outputs which are all included in output intervals.}
	\label{robot_arm_uniform} 
\end{figure}

\begin{figure}[ht!]
	\includegraphics[width=8cm]{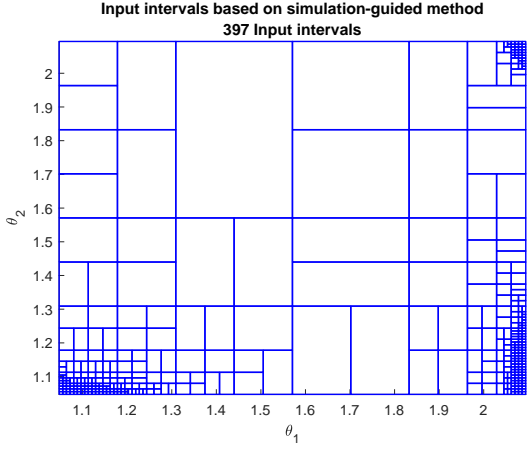}
	\caption{Simulation-guided bisections of input interval by Algorithm \ref{alg1}. Guided by the outputs of simulations, finer partitions are generated when the output intervals are close to boundary of the interval of simulations, and coarse partitions are generated when the output intervals are in the interval of simulations.  }
	\label{robot_arm_input} 
\end{figure}

\section{Reachability Analysis for Neural Network Control Systems}
In this section, we will present the reachability analysis of neural network control systems by incorporating the simulation-guided method with the reachability analysis of ODE models. 
\subsection{Reachability Analysis}
This section presents the reachability analysis and safety verification results for neural network control systems. The developed algorithm combines the aforementioned output set computation result for MLPs and existing reachable set estimation methods for ODE models.

The reachable set estimation for a sampled-data neural network control system in the form of (\ref{sam_closed_loop}) involves two essential portions. First, an over-approximation of the output set of the underlying neural network controllers is supposed to be computed in the employment of the aforementioned output set computation result of neural networks, the Algorithm \ref{alg1}. Then,  the reachable set and output set of the controlled plant (\ref{system}) needs to be computed accordingly.  There are a variety of existing approaches and tools for reachable set computation of systems modeled by ODEs such as those well-developed in \cite{althoff2015introduction,frehse2011spaceex,chen2013flow,duggirala2015c2e2,bak2017hylaa}. Due to the existence of those reachable set estimation of ODE models, we shall not develop new methods or tools for ODE models. We use the following functions to denote the reachable set estimation that is obtained by using reachable set computation tools for sampled data ODE models during $[t_k,t_{k+1}]$
\begin{align}
\mathcal{R}_e([t_{k},t_{k+1}]) & = \texttt{reachODEx}(f,\mathcal{U}(t_k),\mathcal{R}_e(t_k))
\\
\mathcal{Y}_e(t_k) & = \texttt{reachODEy}(h,\mathcal{R}_e(t_k))
\end{align}
where $\mathcal{U}(t_k)$ is the input set for sampling interval $[t_k,t_{k+1}]$. $\mathcal{R}_e(t_{k})$ and $\mathcal{R}_e([t_{k},t_{k+1}])$ are the estimated reachable sets for state $\mathbf{x}(t)$ at sampling instant $t_k$ and interval $[t_k,t_{k+1}]$, respectively. $\mathcal{Y}_e(t_k)$ is the estimated reachable set for output $\mathbf{y}(t_k)$. 

Combining $\texttt{reachODEx}$, $\texttt{reachODEy}$ with $\texttt{reachMLP}$ proposed by Algorithms \ref{alg1}, an over-approximation of the reachable set of a closed-loop system in the form of (\ref{sam_closed_loop}) can be obtained. The computation process is a recursive algorithm is summarized by Algorithm \ref{alg2} and Proposition \ref{pro1}. The general steps can be illustrated as below:
\begin{itemize}
	\item \textbf{Reachable Set Estimation of Neural Network Controller.} Compute the output reachable set estimation for the neural network controller using Algorithm \ref{alg1} at each beginning sampling instant $t_k$, by which an over-approximation of the output  set is obtained. 
	\item \textbf{Reachable Set Estimation of Plant.} As the output generated by the neural network controller holds its value unchanged in $[t_k,t_{k+1}]$, perform the reachable set estimation for the nonlinear continuous-time system using sophisticated methods or tools such as \cite{althoff2015introduction,frehse2011spaceex,chen2013flow,duggirala2015c2e2,bak2017hylaa}.
	\item \textbf{Return for Next Sampling Interval Computation.} Return to the first step of reachable set estimation of neural network controller for the next sampling period  $[t_{k+1},t_{k+2}]$.  
\end{itemize}

\begin{algorithm}[ht!]
	\SetAlgoLined
	\SetKwInOut{Input}{Input}
	\SetKwInOut{Output}{Output}
	\SetKw{Return}{return}
	\Input{System dynamics $f$, $h$; Feedforward neural network $\Phi$
		; Initial Set $\mathcal{X}_0$; Input set $\mathcal{V}$; 
		Tolerance $\varepsilon$; Number of simulations $N$; Sampling sequence $t_k$, $k = 0,1,\ldots,K$; Termination time $t_f$.}
	\Output{Reachable set estimation $\mathcal{R}_e([t_0,t_f])$.}
	\Fn{\texttt{reachNNCS}}{
		\tcc{Initialization}
		$k \gets 0$ \;
		$t_{K+1}\gets t_f$ \;
		$\mathcal{R}_e(t_0) \gets \mathcal{X}_0$ \;
		\tcc{Iteration for all sampling intervals}
		\While{$k \leq K$}{
			$\mathcal{Y}_e(t_k) \gets \texttt{reachODEy}(h,\mathcal{R}_e(t_k))$ \;
			$\mathcal{H} \gets \mathcal{Y}_e(t_k) \times \mathcal{V}$ \;
			$\mathcal{U}_e(t_k) \gets \texttt{reachMLF}(\Phi,\mathcal{H},\varepsilon,N)$ 
			\tcp*{Algorithm \ref{alg1}}
			$\mathcal{R}_e([t_k,t_{k+1}]) \gets \texttt{reachODEx}(f,\mathcal{U}_e,\mathcal{R}_e(t_k))$ \;
			$k \gets k+1$;
		}
		\Return{$\mathcal{R}_e([t_0,t_f]) \gets \bigcup_{k=0,1\ldots,K}\mathcal{R}_e([t_k,t_{k+1}])$}  
	}
	\caption{Reachable Set Estimation for Sampled-Data Neural Network Control Systems} \label{alg2}
\end{algorithm}

\begin{proposition}\label{pro1}
	Given a neural network control system in the form of (\ref{sam_closed_loop}), an initial set $\mathcal{X}_0$ and an input set $\mathcal{V}$, an estimated reachable set  $\mathcal{R}_e([t_0,t_f])$ can be obtained by by Algorithm \ref{alg2} such that  $\mathcal{R}([t_0,t_f])\subseteq \mathcal{R}_e([t_0,t_f])$, where $\mathcal{R}([t_0,t_f])$ is the reachable set of system (\ref{sam_closed_loop}).
\end{proposition}

The safety specification can be examined by checking the emptiness of the intersection between the proposed unsafe regions $\neg \mathcal{S}$ and the reachable set estimation outcome produced by Algorithm \ref{alg2}. According to Lemma \ref{lemma1}, the following result regarding the reachability-based safety verification can be obtained.

\begin{proposition}\label{pro2}
	Given a neural network control system in the form of (\ref{sam_closed_loop}) and a safety specification $\mathcal{S}$, if $ \mathcal{R}_e([t_0,t_f]) \cap \neg \mathcal{S} = \emptyset$, where $\mathcal{R}_e([t_0,t_f])$ is a reachable set estimation obtained by Algorithm \ref{alg2}, then the closed-loop system (\ref{sam_closed_loop}) is safe over time interval $[t_0,t_f]$.
\end{proposition}

\subsection{Safety Verification of Adaptive Cruise Control Systems}
\begin{figure}
	\includegraphics[width=8.5cm]{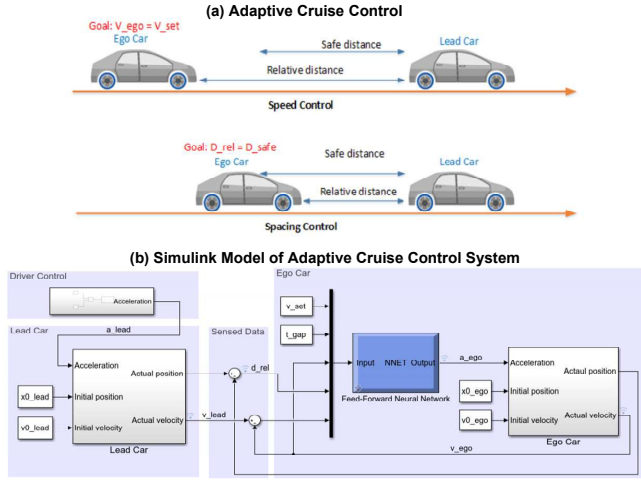}
	\caption{\boldmath  Illustration of adaptive cruise control systems and simulink block diagram of the closed-loop system.}
	\label{acc_sys} 
\end{figure}

In this section,  our approach will be evaluated by the safety verification of an adaptive cruise control (ACC) system equipped with a neural network controller as
depicted in Fig. \ref{acc_sys}. The ACC system consists of two cars, the ego car with ACC module that has a radar sensor to measure
the distance to the lead car which is denoted by $d_{\mathrm{rel}}$, and the relative velocity
against the lead car denoted by $v_{\mathrm{rel}}$. There are two system operating modes including speed control and spacing control. In speed control mode, the ego car travels at a speed $v_{\mathrm{set}}$. In spacing control mode, the ego car's safety control goal is to maintain a safe distance from the leading car, $d_{\mathrm{safe}}$. If $d_{\mathrm{rel}} \geq d_{\mathrm{safe}}$, then speed control mode is active. Otherwise, the spacing control mode is active. 
In summary, the system dynamics is  in the form of
\begin{align}\label{acc}
\left\{ {\begin{array}{*{20}l}
	\dot x_{l}(t) = v_{l}(t)\\
	\dot v_{l}(t) = \gamma_{l}(t)\\
	\dot \gamma_{l}(t) = -2\gamma_{l}(t) + 2 \alpha_{l}(t) - \mu v^2_{l}(t)
	\\
	\dot x_{e} = v_{e}(t)
	\\
	\dot v_{e}(t) = \gamma_e(t)
	\\
	\dot \gamma_{e}(t) = -2\gamma_{e}(t)+2\alpha_{e}(t)-\mu v^{2}_{e}(t)
	\end{array} } \right.
\end{align}
where $x_l$($x_e$), $v_l$($v_e$) and $\gamma_l$($\gamma_e$) are the position, velocity and actual acceleration of the lead (ego) car, respectively.  $\alpha_l$($\alpha_e$) is the acceleration
control input applied to the lead (ego) car, and $\mu = 0.001$ is the friction parameter. The ACC controller we considered here is a $2 \times 20$ feedforward neural network with $\texttt{tanh}$ as its activation functions. The sampling scheme is considered as a periodic sampling every 0.2 seconds, that is $t_{k+1}-t_k = 0.2$ seconds. 

\begin{figure}
	\includegraphics[width=9cm]{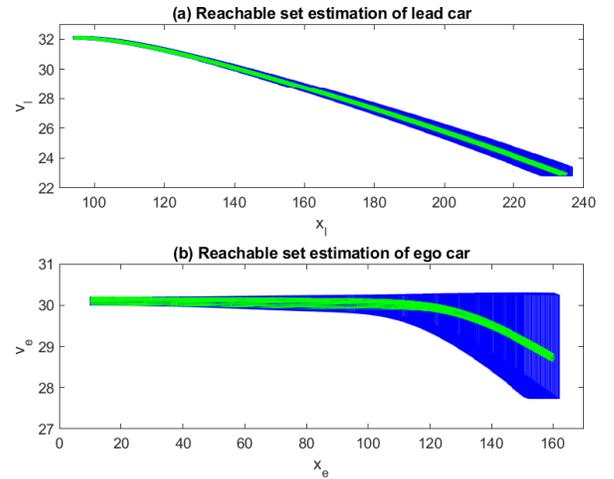}
	\caption{\boldmath Reachable set estimation for both lead car (a) and ego car (b). The over-approximation of reachable set (blue boxes) includes all 100 randomly generated system trajectories (green lines). }
	\label{acc_reach} 
\end{figure}

\begin{figure}
	\includegraphics[width=9cm]{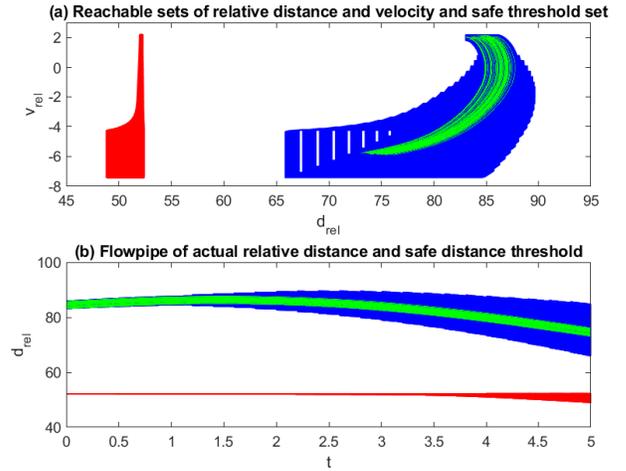}
	\caption{\boldmath Reachable set estimation for relative distance and velocity between lead and ego cars (blue boxes), there is no intersection between relative distance set and safe distance threshold (red boxes) in (a). In (b), the flowpipe of relative distance (blue) has no intersection with the safe distance threshold area (red) which also implies the safety of the ACC system. The green lines are 100 randomly generated system trajectories. }
	\label{acc_rel} 
\end{figure}

The inputs to the neural network ACC control module are:
\begin{itemize}
	\item Driver-set velocity $v_{\mathrm{set}}$;
	\item Time gap $t_{\mathrm{gap}}$;
	\item Velocity of the ego car $v_e$;
	\item Relative distance to the lead car $d_{\mathrm{rel}} = x_l - x_e$;
	\item Relative velocity to the lead car  $v_{\mathrm{rel}} = v_l - v_e$.
\end{itemize}

The output for the  neural network ACC controller is the acceleration of the ego car, $\alpha_e$. In summary, the sampled-data neural network controller for the acceleration control of the ego car is in the form of
\begin{align*}\label{sam_acc}
\alpha_{e}(t) = \Phi(v_{\mathrm{set}}(t_k),t_{\mathrm{gap}},v_{e}(t_k), d_{\mathrm{rel}}(t_k),&v_{\mathrm{rel}}(t_k)),
t \in [t_k,t_{k+1}].
\end{align*}

The threshold of the safe distance between the lead car and the ego car can be considered as a function of the ego car velocity $v_e$. The safety specification is defined as 
\begin{equation} \label{acc_safe}
d_{\mathrm{safe}} > d_{\mathrm{thold}} = d_{\mathrm{def}} + t_{\mathrm{gap}}\cdot v_{e}
\end{equation}
where $d_{\mathrm{def}}$ is the standstill default spacing and $t_{\mathrm{gap}}$ is the time gap between the vehicles.
The safety verification scenario  is that the lead car decelerates with $\alpha_l = -2$ to reduce its speed as an emergency braking occurs. We expect that the ego car is able to maintain a safe relative distance to the lead car to avoid collision. The safety specification is defined by (\ref{acc_safe}) with $t_{\mathrm{gap}} = 1.4$ seconds and $d_{\mathrm{def}}=10$. The time horizon that we want to verify is 5 seconds after the emergency braking comes into play. The initial intervals are $[x_l(0)] = [94, 96]$, $[v_l(0)] = [30, 30.2]$, $[\gamma_l(0)] = 0$, $[x_e(0)] = [10, 11]$, $[v_e(0)] = [30, 30.2]$, $[\gamma_e(0)] = 0$.

We apply Algorithm \ref{alg2} to perform the reachable set estimation for the closed-loop system. The tolerance is chosen as $\varepsilon = 0.1$ and number of simulations is $1 \times 10^5$. For this neural network controller, we use simulation-guided method to compute the output set of the control signal. Meanwhile, for the continuous-time nonlinear dynamics, we use CORA \cite{althoff2015introduction} to do the reachability analysis for the time interval between two sampling instants. The reachable set estimations for both lead car and ego car are shown in Fig. \ref{acc_reach}. In order to verify the safety property, we compute the reachable set estimation of relative distance based on the reachable sets of the lead car and ego car. In Fig. \ref{acc_rel}, the reachable set of relative distance does not violate the threshold of safe distance which is defined by (\ref{acc_safe}), so it can be concluded that the ACC system is safe during the time interval $[0,5]$ seconds in this safety verification scenario of interest\footnote{The  source code is available at:
	https://github.com/xiangweiming/ignnv}.  

\section{Conclusion}
This work investigated the reachable set estimation and safety verification problems for a class of neural network control systems which can be modeled as sampled data continuous-time dynamical systems. A novel simulation-guided approach is developed to soundly over-approximate the output set of a class of feedforward neural networks called MLP. Based on the interval analysis of neural networks and guidance of simulations generated from neural networks, the output reachable set can be efficiently over-approximated upon avoidance of unnecessary computation cost. Compared with the other simulation-based approach in \cite{xiang2018output}, the approach developed in this paper can reduce the computational cost significantly. In the robotic arm example, it only needs 3\% computational cost of the method in \cite{xiang2018output} for the same interval readability analysis results. 
Furthermore, in a combination of reachable set computation methods and tools for ODE models, a recursive algorithm is developed to perform reachable set estimation and safety verification of neural network control systems.  Beyond the initial results derived in this work,   other modeling and reachability analysis approaches for the plant and neural network controllers, as well as broader classes of neural networks, should be considered in the future study. For example, though our approach is general in the sense that it is not particularly designed for any specific activation functions, the simulation-guided idea has the potential to be further applied to other methods dealing with specific activation functions such as ReLU neural networks to enhance their scalability.

\section*{Acknowledgment}
The  material  presented  in  this  paper  is  based  upon  work supported  by  the  Air  Force  Office  of  Scientific  Research (AFOSR)  through  contract  number  FA9550-18-1-0122  and the Defense Advanced Research Projects Agency (DARPA) through contract number FA8750-18-C-0089. The U.S. Government is authorized to reproduce and distribute reprints for  Government  purposes  notwithstanding  any  copyright notation  thereon.  The  views  and  conclusions  contained herein are those of the authors and should not be interpreted as  necessarily  representing  the  official  policies  or  endorsements, either expressed or implied, of AFOSR or DARPA. This work is also partially supported by Research Scholarship and Creative Activity Grant Program of Augusta University.

\ifCLASSOPTIONcaptionsoff
  \newpage
\fi



\bibliographystyle{IEEEtran}
\bibliography{reference}

\begin{thebibliography}{10}
\providecommand{\url}[1]{#1}
\csname url@samestyle\endcsname
\providecommand{\newblock}{\relax}
\providecommand{\bibinfo}[2]{#2}
\providecommand{\BIBentrySTDinterwordspacing}{\spaceskip=0pt\relax}
\providecommand{\BIBentryALTinterwordstretchfactor}{4}
\providecommand{\BIBentryALTinterwordspacing}{\spaceskip=\fontdimen2\font plus
\BIBentryALTinterwordstretchfactor\fontdimen3\font minus
  \fontdimen4\font\relax}
\providecommand{\BIBforeignlanguage}[2]{{%
\expandafter\ifx\csname l@#1\endcsname\relax
\typeout{** WARNING: IEEEtran.bst: No hyphenation pattern has been}%
\typeout{** loaded for the language `#1'. Using the pattern for}%
\typeout{** the default language instead.}%
\else
\language=\csname l@#1\endcsname
\fi
#2}}
\providecommand{\BIBdecl}{\relax}
\BIBdecl

\bibitem{wu2014exponential}
Z.-G. Wu, P.~Shi, H.~Su, and J.~Chu, ``Exponential stabilization for
  sampled-data neural-network-based control systems,'' \emph{IEEE Transactions
  on Neural Networks and Learning Systems}, vol.~25, no.~12, pp. 2180--2190,
  2014.

\bibitem{yu1998stable}
S.-H. Yu and A.~M. Annaswamy, ``Stable neural controllers for nonlinear dynamic
  systems,'' \emph{Automatica}, vol.~34, no.~5, pp. 641--650, 1998.

\bibitem{ge1999adaptive}
S.~S. Ge, C.~C. Hang, and T.~Zhang, ``Adaptive neural network control of
  nonlinear systems by state and output feedback,'' \emph{IEEE Transactions on
  Systems, Man, and Cybernetics, Part B (Cybernetics)}, vol.~29, no.~6, pp.
  818--828, 1999.

\bibitem{hunt1992neural}
K.~J. Hunt, D.~Sbarbaro, R.~{\.Z}bikowski, and P.~J. Gawthrop, ``Neural
  networks for control systems: A survey,'' \emph{Automatica}, vol.~28, no.~6,
  pp. 1083--1112, 1992.

\bibitem{julian2017neural}
K.~Julian and M.~J. Kochenderfer, ``Neural network guidance for
  \textsc{UAV}s,'' in \emph{AIAA Guidance Navigation and Control Conference
  (GNC)}, 2017.

\bibitem{bojarski2016end}
M.~Bojarski, D.~Del~Testa \emph{et~al.}, ``End to end learning for self-driving
  cars,'' \emph{arXiv preprint arXiv:1604.07316}, 2016.

\bibitem{julian2016policy}
K.~Julian, J.~Lopez, J.~Brush, M.~Owen, and M.~Kochenderfer, ``Policy
  compression for aircraft collision avoidance systems,'' in \emph{35th Digital
  Avionics Systems Conference (DASC)}, 2016, pp. 1--10.

\bibitem{szegedy2013intriguing}
C.~Szegedy, W.~Zaremba, I.~Sutskever, J.~Bruna, D.~Erhan, I.~Goodfellow, and
  R.~Fergus, ``Intriguing properties of neural networks,'' in
  \emph{International Conference on Learning Representations}, 2014.

\bibitem{katz2017reluplex}
G.~Katz, C.~Barrett, D.~Dill, K.~Julian, and M.~Kochenderfer, ``Reluplex: An
  efficient \textsc{SMT} solver for verifying deep neural networks,'' in
  \emph{International Conference on Computer Aided Verification}.\hskip 1em
  plus 0.5em minus 0.4em\relax Springer, 2017, pp. 97--117.

\bibitem{xiang2018output}
W.~{Xiang}, H.~{Tran}, and T.~T. {Johnson}, ``Output reachable set estimation
  and verification for multilayer neural networks,'' \emph{IEEE Transactions on
  Neural Networks and Learning Systems}, vol.~29, no.~11, pp. 5777--5783, 2018.

\bibitem{pulina2010abstraction}
L.~Pulina and A.~Tacchella, ``An abstraction-refinement approach to
  verification of artificial neural networks,'' in \emph{International
  Conference on Computer Aided Verification}.\hskip 1em plus 0.5em minus
  0.4em\relax Springer, 2010, pp. 243--257.

\bibitem{pulina2012challenging}
------, ``Challenging \textsc{SMT} solvers to verify neural networks,''
  \emph{AI Communications}, vol.~25, no.~2, pp. 117--135, 2012.

\bibitem{huang2017safety}
X.~Huang, M.~Kwiatkowska, S.~Wang, and M.~Wu, ``Safety verification of deep
  neural networks,'' in \emph{International Conference on Computer Aided
  Verification}.\hskip 1em plus 0.5em minus 0.4em\relax Springer, 2017, pp.
  3--29.

\bibitem{xu2016reachable}
Z.~Xu, H.~Su, P.~Shi, R.~Lu, and Z.-G. Wu, ``Reachable set estimation for
  markovian jump neural networks with time-varying delays,'' \emph{IEEE
  Transactions on Cybernetics}, vol.~47, no.~10, pp. 3208--3217, 2016.

\bibitem{zuo2014non}
Z.~Zuo, Z.~Wang, Y.~Chen, and Y.~Wang, ``A non-ellipsoidal reachable set
  estimation for uncertain neural networks with time-varying delay,''
  \emph{Communications in Nonlinear Science and Numerical Simulation}, vol.~19,
  no.~4, pp. 1097--1106, 2014.

\bibitem{thuan2018reachable}
M.~V. Thuan, H.~M. Tran, and H.~Trinh, ``Reachable sets bounding for
  generalized neural networks with interval time-varying delay and bounded
  disturbances,'' \emph{Neural Computing and Applications}, vol.~29, no.~10,
  pp. 783--794, 2018.

\bibitem{lomuscio2017an_arxiv}
A.~Lomuscio and L.~Maganti, ``An approach to reachability analysis for
  feed-forward \textsc{R}e\textsc{LU} neural networks,'' \emph{arXiv preprint
  arXiv:1706.07351}, 2017.

\bibitem{dutta2017output}
S.~Dutta, S.~Jha, S.~Sanakaranarayanan, and A.~Tiwari, ``Output range analysis
  for deep neural networks,'' \emph{arXiv preprint arXiv:1709.09130}, 2017.

\bibitem{ehlers2017formal}
R.~Ehlers, ``Formal verification of piecewise linear feed-forward neural
  networks,'' in \emph{International Symposium on Automated Technology for
  Verification and Analysis}.\hskip 1em plus 0.5em minus 0.4em\relax Springer,
  2017, pp. 269--286.

\bibitem{xiang2017reachable_arxiv}
W.~Xiang, H.-D. Tran, and T.~T. Johnson, ``Reachable set computation and safety
  verification for neural networks with \textsc{R}e\textsc{LU} activations,''
  \emph{arXiv preprint arXiv: 1712.08163}, 2017.

\bibitem{tran2019star}
H.-D. Tran, D.~M. Lopez, P.~Musau, X.~Yang, L.~V. Nguyen, W.~Xiang, and T.~T.
  Johnson, ``Star-based reachability analysis for deep neural networks,'' in
  \emph{International Symposium on Formal Methods}, 2019.

\bibitem{zeng2001sensitivity}
X.~Zeng and D.~S. Yeung, ``Sensitivity analysis of multilayer perceptron to
  input and weight perturbations,'' \emph{IEEE Transactions on Neural
  Networks}, vol.~12, no.~6, pp. 1358--1366, 2001.

\bibitem{zeng2003quantified}
------, ``A quantified sensitivity measure for multilayer perceptron to input
  perturbation,'' \emph{Neural Computation}, vol.~15, no.~1, pp. 183--212,
  2003.

\bibitem{piche1995selection}
S.~W. Piche, ``The selection of weight accuracies for madalines,'' \emph{IEEE
  Transactions on Neural Networks}, vol.~6, no.~2, pp. 432--445, 1995.

\bibitem{xi2013architecture}
X.-Z. Wang, Q.-Y. Shao, Q.~Miao, and J.-H. Zhai, ``Architecture selection for
  networks trained with extreme learning machine using localized generalization
  error model,'' \emph{Neurocomputing}, vol. 102, pp. 3--9, 2013.

\bibitem{shi2005sensitivity}
D.~Shi, D.~S. Yeung, and J.~Gao, ``Sensitivity analysis applied to the
  construction of radial basis function networks,'' \emph{Neural networks},
  vol.~18, no.~7, pp. 951--957, 2005.

\bibitem{xiang2018reachable_acc}
W.~Xiang, H.-D. Tran, J.~A. Rosenfeld, and T.~T. Johnson, ``Reachable set
  estimation and safety verification for piecewise linear systems with neural
  network controllers,'' \emph{arXiv preprint arXiv:1802.06981}, 2018.

\bibitem{xiang2019reachability}
W.~{Xiang}, H.~{Tran}, and T.~T. {Johnson}, ``Reachability analysis and safety
  verification for neural network control systems,'' in \emph{AAAI Spring
  Symposium on Verification of Neural Networks}, 2019.

\bibitem{dutta2019reachability}
S.~Dutta, X.~Chen, and S.~Sankaranarayanan, ``Reachability analysis for neural
  feedback systems using regressive polynomial rule inference,'' in
  \emph{Proceedings of the 22nd ACM International Conference on Hybrid Systems:
  Computation and Control}.\hskip 1em plus 0.5em minus 0.4em\relax New York,
  NY, USA: ACM, 2019, pp. 157--168.

\bibitem{ivanov2019verisig}
R.~Ivanov, J.~Weimer, R.~Alur, G.~J. Pappas, and I.~Lee, ``Verisig: Verifying
  safety properties of hybrid systems with neural network controllers,'' in
  \emph{Proceedings of the 22Nd ACM International Conference on Hybrid Systems:
  Computation and Control}, ser. HSCC '19.\hskip 1em plus 0.5em minus
  0.4em\relax New York, NY, USA: ACM, 2019, pp. 169--178.

\bibitem{hornik1989multilayer}
K.~Hornik, M.~Stinchcombe, and H.~White, ``Multilayer feedforward networks are
  universal approximators,'' \emph{Neural Networks}, vol.~2, no.~5, pp.
  359--366, 1989.

\bibitem{moore2009introduction}
R.~E. Moore, R.~B. Kearfott, and M.~J. Cloud, \emph{Introduction to interval
  analysis}.\hskip 1em plus 0.5em minus 0.4em\relax SIAM, 2009, vol. 110.

\bibitem{skelboe1974computation}
S.~Skelboe, ``Computation of rational interval functions,'' \emph{BIT Numerical
  Mathematics}, vol.~14, no.~1, pp. 87--95, 1974.

\bibitem{althoff2015introduction}
M.~Althoff, ``An introduction to cora 2015,'' in \emph{Workshop on Applied
  Verification for Continuous and Hybrid Systems}, 2015.

\bibitem{frehse2011spaceex}
G.~Frehse, C.~Le~Guernic, A.~Donz{\'e}, S.~Cotton, R.~Ray, O.~Lebeltel,
  R.~Ripado, A.~Girard, T.~Dang, and O.~Maler, ``Spaceex: Scalable verification
  of hybrid systems,'' in \emph{International Conference on Computer Aided
  Verification}.\hskip 1em plus 0.5em minus 0.4em\relax Springer, 2011, pp.
  379--395.

\bibitem{chen2013flow}
X.~Chen, E.~{\'A}brah{\'a}m, and S.~Sankaranarayanan, ``Flow*: An analyzer for
  non-linear hybrid systems,'' in \emph{International Conference on Computer
  Aided Verification}.\hskip 1em plus 0.5em minus 0.4em\relax Springer, 2013,
  pp. 258--263.

\bibitem{duggirala2015c2e2}
P.~S. Duggirala, S.~Mitra, M.~Viswanathan, and M.~Potok, ``\textsc{C2E2}: a
  verification tool for stateflow models,'' in \emph{International Conference
  on Tools and Algorithms for the Construction and Analysis of Systems}.\hskip
  1em plus 0.5em minus 0.4em\relax Springer, 2015, pp. 68--82.

\bibitem{bak2017hylaa}
S.~Bak and P.~S. Duggirala, ``Hy\textsc{LAA}: A tool for computing
  simulation-equivalent reachability for linear systems,'' in \emph{Proceedings
  of the 20th International Conference on Hybrid Systems: Computation and
  Control}.\hskip 1em plus 0.5em minus 0.4em\relax ACM, 2017, pp. 173--178.

\end{thebibliography}
\end{document}